\documentclass[twoside,a4paper]{IEEEtran}
\usepackage[noadjust]{cite} 
\usepackage{lettrine}
\usepackage{multirow}
\usepackage{amsmath}
\usepackage{overpic}

\usepackage{enumitem}
\usepackage{amsopn}
\usepackage{amssymb}
\usepackage{accents}
\usepackage{mathrsfs}
\usepackage{pifont}
\newcommand{\crh}[1]{{\color{black}#1}}
\newcommand{\snr}{\mathsf{snr}}
\newcommand{\myvec}[1]{{\mathbf{#1}}}
\newcommand{\review}[1]{{\color{black}#1}}
\usepackage{bm}         
\usepackage{fix-cm}     

\usepackage{mleftright} 
\mleftright                 

\usepackage{mathdots}   

\makeatletter
\def\IEEElabelanchoreqn#1{\bgroup
	\def\@currentlabel{\p@equation\theequation}\relax
	\def\@currentHref{\@IEEEtheHrefequation}\label{#1}\relax
	\Hy@raisedlink{\hyper@anchorstart{\@currentHref}}\relax
	\Hy@raisedlink{\hyper@anchorend}\egroup}
\makeatother

\usepackage[utf8]{inputenc}
\usepackage[T1]{fontenc} 

\usepackage{tikz}
\usepackage{tikz}
\allowdisplaybreaks
\usetikzlibrary{datavisualization}

\usetikzlibrary{positioning, fit, calc} 
\tikzset{block/.style={draw, thick, minimum width=0.5cm, minimum height=0.5cm, align=center}, 
	line/.style={-latex}   
}

\usepackage{pgfplots}

\usepackage{subfigure}
\usepackage{algorithmicx}
\usepackage{algorithm}
\usepackage{algpseudocode}
\usepackage{amsthm}
\usepackage{booktabs}
\usepackage{balance}

\usepgfplotslibrary{units}
\usetikzlibrary{spy,backgrounds}
\usetikzlibrary{decorations.markings}
\usetikzlibrary{positioning}
\usetikzlibrary{shadows} 
\usetikzlibrary{calc}
\usetikzlibrary{patterns}
\usepackage{pgfplots}
\usepgfplotslibrary{units}
\usetikzlibrary{spy,backgrounds}
\usetikzlibrary{decorations.markings}
\usepackage{pgfplotstable}

\newtheorem{lemma}{Lemma}

\usepackage{hyphenat} 

\newcommand{\midk}[1]{\kern0.1em #1 \kern0.1em}
\newcommand{\middlek}[1]{\kern0.1em \middle#1 \kern0.1em}
\newcommand{\bigk}[1]{\kern-0.1em \bigm#1 \kern-0.1em}
\newcommand{\Bigk}[1]{\kern-0.1em \Bigm#1 \kern-0.1em}
\newcommand{\biggk}[1]{\kern-0.1em \biggm#1 \kern-0.1em}
\newcommand{\Biggk}[1]{\kern-0.1em \Biggm#1 \kern-0.1em}

\newcommand{\tn}[1]{\textnormal{#1}}

\newcommand{\mat}[1]{\mathsf{#1}} 
\newcommand{\trans}[1]{#1^{\textnormal{\textsf{\tiny T}}}} 

\newcommand{\hermi}[1]{#1^{\dagger}} 

\newcommand{\const}[1]{\textnormal{\usefont{U}{eur}{m}{n}\selectfont #1}} 
\newcommand{\relD}{\mathop{}\!\mathsf{D}}         
\newcommand{\relDf}[2]{\relD\left(#1 \kern0.1em\middle\|\kern0.1em #2\right)}
\newcommand{\erelDf}[2]{\relD(#1 \kern0.1em\|\kern0.1em #2)}
\newcommand{\bigrelDf}[2]{\relD\bigl(#1 \kern-0.1em \bigm\| \kern-0.1em#2\bigr)}
\newcommand{\BigrelDf}[2]{\relD\Bigl(#1 \kern-0.1em \Bigm\| \kern-0.1em#2\Bigr)}
\newcommand{\biggrelDf}[2]{\relD\biggl(#1 \kern-0.1em \biggm\| \kern-0.1em#2\biggr)}
\newcommand{\BiggrelDf}[2]{\relD\Biggl(#1 \kern-0.1em \Biggm\| \kern-0.1em#2\Biggr)}

\newcommand{\nr}{n_{\tn{R}}}

\newcommand{\Prvcond}[2]{\Pr\left[#1 \kern0.1em\middle|\kern0.1em #2\right]}
\newcommand{\ePrvcond}[2]{\Pr[#1 \kern0.1em|\kern0.1em #2]}
\newcommand{\bigPrvcond}[2]{\Pr\bigl[#1 \kern-0.1em \bigm| \kern-0.1em#2\bigr]}
\newcommand{\BigPrvcond}[2]{\Pr\Bigl[#1 \kern-0.1em \Bigm| \kern-0.1em#2\Bigr]}
\newcommand{\biggPrvcond}[2]{\Pr\biggl[#1 \kern-0.1em \biggm| \kern-0.1em#2\biggr]}
\newcommand{\BiggPrvcond}[2]{\Pr\Biggl[#1 \kern-0.1em \Biggm| \kern-0.1em#2\Biggr]}

\newcommand{\Prscond}[2]{\Pr\left(#1 \kern0.1em\middle|\kern0.1em #2\right)}
\newcommand{\ePrscond}[2]{\Pr(#1 \kern0.1em|\kern0.1em #2)}
\newcommand{\bigPrscond}[2]{\Pr\bigl(#1 \kern-0.1em \bigm| \kern-0.1em#2\bigr)}
\newcommand{\BigPrscond}[2]{\Pr\Bigl(#1 \kern-0.1em \Bigm| \kern-0.1em#2\Bigr)}
\newcommand{\biggPrscond}[2]{\Pr\biggl(#1 \kern-0.1em \biggm| \kern-0.1em#2\biggr)}
\newcommand{\BiggPrscond}[2]{\Pr\Biggl(#1 \kern-0.1em \Biggm| \kern-0.1em#2\Biggr)}

\newcommand{\Exp}{\operatorname{\textnormal{\textsf{E}}}}

\newcommand{\Econd}[3][]{\Exp_{#1}\left[#2 \kern0.1em\middle|\kern0.1em #3\right]}
\newcommand{\eEcond}[3][]{\Exp_{#1}[#2 \kern0.1em|\kern0.1em #3]}
\newcommand{\bigEcond}[3][]{\Exp_{#1}\bigl[#2 \kern-0.1em \bigm| \kern-0.1em #3\bigr]}
\newcommand{\BigEcond}[3][]{\Exp_{#1}\Bigl[#2 \kern-0.1em \Bigm| \kern-0.1em #3\Bigr]}
\newcommand{\biggEcond}[3][]{\Exp_{#1}\biggl[#2 \kern-0.1em \biggm| \kern-0.1em #3\biggr]}
\newcommand{\BiggEcond}[3][]{\Exp_{#1}\Biggl[#2 \kern-0.1em \Biggm| \kern-0.1em #3\Biggr]}

\makeatletter

\newcommand{\Rmnum}[1]{\expandafter\@slowromancap\romannumeral #1@}
\makeatother

\usepackage{threeparttable}

\usepackage{diagbox}
\usepackage{array}
\usepackage{boldline}
\usepackage{dsfont} 
\usepackage{algorithmicx}
\usepackage{algpseudocode}

\allowdisplaybreaks
\usetikzlibrary{datavisualization}

\usetikzlibrary{spy,backgrounds}
\usetikzlibrary{decorations.markings}

\usepackage{xcolor}
\usepackage{booktabs}
\begin{document}
	
	\title{Grouped Annulus-Modulated Transceiver Is Almost Full DoF-Achieving for RIS-Assisted Symbiotic Radios Over Spatial-Correlated Channels}
	
	\author{Ruo-Qi Sun, Jianfeng Shi, Yonggang Zhu, Mingliang Xie, Kang Luo, Yifu Sun, Ru-Han Chen$^{\ast}$, Kang An$^{\ast}$ 
		\thanks{This work is supported by Research Program of National University of Defense Technology under Grant No. ZK23-57, and the National Natural Science Foundation of China under Grant No. 62401593. R.-Q. Sun is with School of Electronics and Information Engineering, Nanjing University of Information Science and Technology, Nanjing, China, and also with Sixty-Third Research Institute, National University of Defense Technology, Nanjing, China (e-mail: 202412491635@nuist.edu.cn). J. Shi	is with School of Electronics and Information Engineering, Nanjing University of Information Science and Technology, Nanjing, China (e-mail: jianfeng.shi@nuist.edu.cn).
		Y. Zhu, Y. Sun, R.-H. Chen and K. An are with Sixty-Third Research Institute, National University of Defense Technology, Nanjing, China (e-mail: zhumaka1982@163.com, sunyifu\_nudt@163.com, tx\_rhc22@nudt.edu.cn, ankang89@nudt.edu.cn). M. Xie and K. Luo are with National Key Laboratory of Electromagnetic Energy,	Naval University of Engineering, Wuhan, China (18627191916@163.com, luokangemc@nue.edu.cn).  \textit{(Corresponding Authors: Ru-Han Chen and Kang An.)}}}

	\maketitle
	
	\begin{abstract}
	This paper considers a RIS-assisted symbiotic communication system, where additional information is conveyed by the passive reconfigurable intelligent surface (RIS).
	In existing schemes, individual phase modulation is usually adopted at the RIS elements, which severely limits exploiting all extra multiplexing gains brought by the RIS.
	To address the issue, we propose a novel matrix decomposition algorithm that transforms the equivalent channel into a structured form while effectively suppressing the decomposition residual. 
	Based on this, a novel transceiver architecture employing grouped annulus modulation (GAM) with a hexagonal-lattice-based constellation is developed, which is capable of achieving the full degrees of freedom (DoFs) when the decomposition algorithm performs as expected.
	Numerical results demonstrate that the proposed transceiver achieves much higher communication rates, thereby leading to higher spectral efficiency, compared to the conventional phase-only modulation scheme, while maintaining comparable error performance.
	\end{abstract}

\section{Introduction}\label{Sec.introduction}
\lettrine[lines=2]{F}{UTURE} 
wireless networks require ultra-high data rates, improved energy efficiency, and massive connectivity, which has motivated extensive research on key enabling technologies, such as millimeter-wave (mmWave) communications and intelligent electromagnetic materials.
These emerging technologies also impose new requirements on the design of wireless transceiver architectures. 
In this context, reconfigurable intelligent surfaces (RISs), as a promising technology, have attracted significant attention due to their capability of programmably controlling the electromagnetic propagation environment \cite{Basar2024VTM,Wu2021Tcom}, and are expected to significantly enhance system performance while providing a new paradigm for next-generation transceiver design.

A natural use of the RIS in wireless communications is to employ it as a passive beamformer, which aims to maximize the received signal-to-noise ratio (SNR) by appropriately adjusting the phase shifts of its controllable elements \cite{Wu2019TWC,Zhao2025TMC,Xue2026TNCE,Liu2024TWC,Ge2025TWC}.
A large body of literature has been devoted to the joint optimization of the beamforming vectors at the RIS, the transmitter, and the receiver for various scenarios, e.g.,  single-user \cite{chen2025TWC} and multi-user multiple-input single-output (MISO) systems \cite{zhao2025Tcom}, downlink multi-user systems \cite{Liu2025TWC}, multiple-input and multiple-output (MIMO) systems \cite{Wang2024TWC}.

Recently, a more attractive use of the RIS for conveying additional information has been proposed under the concept of \textit{symbiotic radios}, also referred to as \textit{symbiotic active/passive transmission} or \textit{simultaneous active and passive information transfer}, which aims to enable mutually beneficial information exchange between active and passive transmissions via shared spectrum, energy, and infrastructure \cite{Liang2020TCCN}. 
From an information-theoretic perspective, one key motivation lies in the fact that employing the RIS purely as a passive beamformer cannot fully exploit the available multiplexing gains, except for the special case of rank-one channels, as revealed in \cite{Hei2024TIT} and our previous work \cite{chen2024twc}.
In RIS-assisted symbiotic communication systems, the phase shifts of the RIS are not only determined by the channel state but can also be jointly exploited with the active transmit antenna to convey information \cite{Li2023IOT,Jiang2023TSP}. When the RIS has data to transmit, it can modulate information by adjusting its phase shifts, while simultaneously enhancing the channel gain of the primary transmission link, thereby improving the overall system capacity \cite{Zhang2024TCCN}.
In addition, such passive transmission employs ultra-low-power devices and conveys information by reflecting electromagnetic signals from the active transmission link, thus avoiding the need for dedicated radio-frequency (RF) chains and additional spectrum resources \cite{Liang2020TCCN}. Based on this mechanism, the base station can be equipped with only a single active transmit antenna together with an RIS comprising a large number of reflecting elements. This enables the construction of a virtual large-scale MIMO system and a single-RF-chain MIMO architecture \cite{Li2021WC,Karasik2021ISIT}.
This architecture significantly reduces hardware complexity and energy consumption, while simplifying digital signal processing, thus providing an energy-efficient solution for future wireless communication systems. However, despite these advantages, existing RIS-assisted symbiotic communication schemes are still not sufficiently efficient, which calls for further investigation into more advanced transceiver designs.

\subsection{Related Works}
In recent years, transceiver design for symbiotic radio systems have been extensively studied in the literature from two main aspects.

\subsubsection{Characterization of Achievable Rates}
A central problem in symbiotic radio systems is the characterization of information-theoretic performance metrics, including bounds on channel capacities or degrees of freedom (DoFs) (i.e., the limiting pre-log factor of the achievable rate). 
It is shown in \cite{Seddik2022CL} that phase modulation at RIS elements can bring certain DoF gains for MIMO systems even under non-coherent detection. 
In \cite{Hei2024TIT}, the maximum DoF of symbiotic radios is characterized when information is jointly modulated on the transmitted signal and the RIS phase shifts.
Furthermore, the achievable rate of such joint transmission schemes is investigated under specific modulation formats in \cite{Karasik2021Tcom}.
\subsubsection{Practical Modulation Designs}
In \cite{Yao2024TWC,Yao2023TVT}, a superimposed modulation framework is developed, where information is embedded into a pre-optimized RIS phase pattern. 
To address the challenges arising from the high-dimensional constant-modulus constraint and reduce implementation complexity, a spatial sigma-delta modulation scheme is proposed in \cite{Keung2024ICASSP}, where a carefully designed phase sequence is employed at the RIS elements to achieve bit error rate (BER) performance comparable to that of the unquantized zero-forcing (ZF) scheme, while maintaining low computational complexity.	

Spatial modulation, or more generally index modulation, is investigated in \cite{Lin2021TWC,Basar2020Tcom}. In \cite{Lin2022TWC}, a quadrature reflection modulation scheme is proposed, where the RIS elements are partitioned into two subsets and information is conveyed via the index of the activated subset.
Furthermore, by properly adjusting the phase shifts at the RIS, information can also be conveyed through the receive antenna index modulation or via dual-ring phase shift keying (PSK) constellations, as shown in \cite{Lei2021WC,Wu2021TVT}.

Several studies focus on efficient joint coding and decoding strategies. In \cite{Yan2020JSAC}, a passive beamforming and information transmission architecture is proposed, which is further generalized in \cite{Jiang2023TSP}. In \cite{Li2026ITJ}, the authors propose an XOR-based space-time joint coding framework, which significantly improves the SNR and error performance of an RIS-assisted symbiotic communication system. 
In our previous work \cite{chen2024twc}, based on QR decomposition and successive interference cancellation (SIC), we design a transceiver architecture for the RIS-assisted single-input multiple-output (SIMO) system, which is shown to achieve full DoFs under spatially independent channel conditions.

\subsection{Motivation and Contributions}
It should be noted that most existing designs for RIS-assisted communication systems do not take the spatial correlation \review{\cite{Sun2021WCL}} among \review{phase-controllable} elements into account, which can significantly reduce achievable DoF \cite{philipp2026TAP}. 
Moreover, existing transceiver and modulation designs are generally not tailored to effectively exploit the structured characteristics of spatially correlated channels, making it difficult to fully recover the potential DoFs of RIS-assisted systems.Thus, \textit{how to achieve all potential DoFs of RIS-assisted system in the spatially-correlated regime} is a practical but fundamental problem. 
Addressing this challenge is indispensable for accelerating the design, deployment, and optimization of RIS-assisted symbiotic communication systems, which calls for a systematic transceiver design framework despite its challenging nature. 
To the best of the authors' knowledge, there has been no relevant research so far.

In this paper, we propose a novel transceiver architecture as well as a simple constellation design for the RIS-assisted symbiotic communication system, by which full DoFs can be achieved. The main contributions are summarized as follows:
	
	\begin{enumerate}
		\item \textit{Novel Transceiver Architecture:}
		Motivated by the expanded constellation geometry of the sum of multiple phase-modulated signals (see Lemma~\ref{lem:achievable_magnitude}), we propose a novel transceiver architecture that aims to convert the equivalent channel matrix into the one of a special row-echelon form via a unitary transformation, where the column indices of the first nonzero elements \review{in two adjacent rows} differ by $2$. 
		Theoretical analysis shows that, combined with SIC, this transceiver can achieve full DoFs if the matrix decomposition algorithm (used for searching the required unitary matrix) is well-designed.
	
		\item \textit{Low-Complexity Matrix Decomposition:} 
		To obtain the decomposition of the equivalent channel matrix with the residual error (compared with the ideal row-echelon form) being as small as possible, we propose a simple algorithm that computes a closed-form unit vector with least residual in each iteration. 
		Numerical results show the superiority of the proposed algorithm especially in the case where strong spatial correlation exists.
		
		\item \textit{Annular Constellation for Subchannels:}  
		With the aforementioned matrix decomposition method and SIC, the proposed transceiver divides the equivalent channel into a beamforming subchannel and several phase-modulated subchannels, each of which has two phase inputs thereby forming annular constellation geometry. 
		For these phase-modulated subchannels, we design a two-dimensional constellation from the hexagonal lattice, which perfectly fits the annular constellation geometry. 
		The fast methods for computing the cardinality of the annular constellation and decomposing the constellation points into two phase input signals are also given.
		Numerical results show that the proposed transceiver has significantly higher spectral efficiency than the existing transceiver architecture under the condition of comparable error performance.
		
	\end{enumerate}

\subsection{Organization and Notations}
\review{
The remainder of this paper is organized as follows.
Sec.~\ref{Sec.System Model} presents the system model with a special emphasis on the spatial correlation among the channel gains.
In Sec.~\ref{sec:transceiver_architecture}, we introduce the considered transceiver architecture based on unitary transformation and develop a transceiver architecture based on unitary rotation and SIC.
In Sec.~\ref{sec:GAM}, we propose a novel transceiver architecture as well as a low-complexity matrix decomposition algorithm, which is shown to achieve full DoFs under well-designed conditions.
Sec.~\ref{Sec.Annular Constellation} is devoted to the constellation design problem involved in our proposed transceiver, and fast algorithms for enumerating and decomposing constellation points are also given. 
Numerical results are provided in Sec.~\ref{Sec.Simulation result} to validate the superiority of our proposed transceiver architecture. In Sec.~\ref{Sec.Conclusion}, we conclude the paper and discuss some open problems.

\textit{Notations:} Unless otherwise specified, we use lower-case, boldface lower-case, upper-case, and boldface upper-case letters to denote deterministic scalars, deterministic vectors, random variables, and random vectors, respectively. For a complex number $s\in \mathbb{C}$, its angle and modulus are denoted by $\angle s$ and $\left| s \right|$, respectively. For a complex-valued vector $\myvec{x}$, $\trans{\myvec{x}}$, $\myvec{x}^{\dagger}$, $\left\|
\myvec{x} \right\|_2$, and $\mathsf{diag}(\myvec{x})$ denote its transpose, conjugate transpose, $\ell_2$-norm, and the diagonal matrix with diagonal elements given by $\myvec{x}$, respectively.
The expectation operator is denoted by $\mathbb{E}[\cdot]$.
The notation $\mathcal{CN}(\cdot,\cdot)$ denotes the circularly symmetric complex Gaussian distribution.
}

\section{System Model} \label{Sec.System Model}
\subsection{Original Model}
\review{In line with} \cite{chen2024twc}, we consider a single-input multiple-output wireless link with a single transmit antenna and $\nr$ receive antennas, where an RIS with $n$ reflective elements is deployed and a direct path between the transmitter and the receiver is allowed. See Fig.~\ref{fig.RIS_SIMO} for the system diagram.

\begin{figure}[!htp]
	\centering
	\includegraphics[width=0.95\columnwidth]{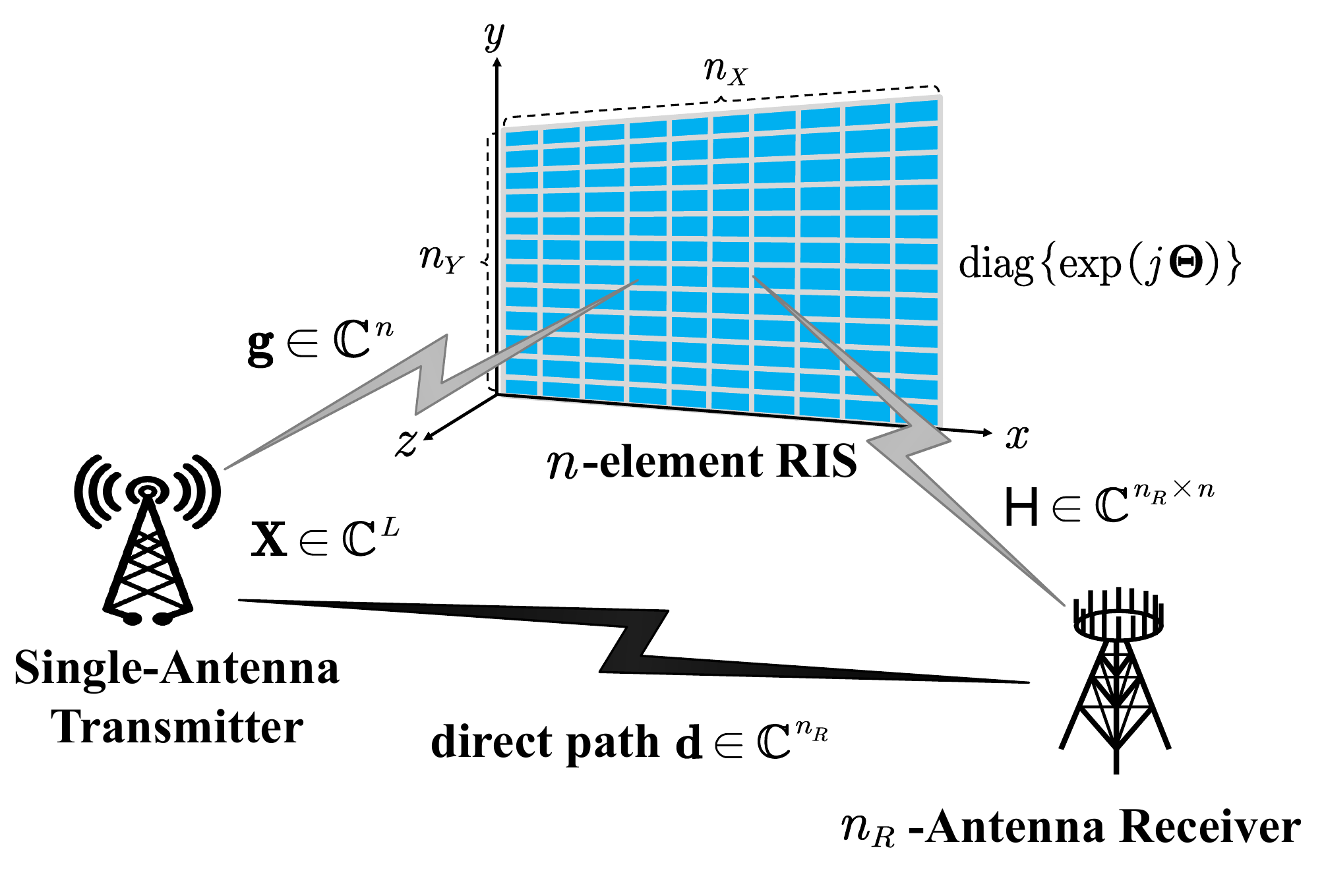}
	\caption{Diagram of the considered RIS-aided SIMO wireless communication system.}
	\label{fig.RIS_SIMO}
\end{figure}

Taking the limited re-configuration rate of the RIS into account, we denote the ratio of the symbol rate at the transmitter side and the re-configuration rate of the RIS by $L\in \mathbb{N}_{+}$. Then the received signal in a length-$L$ sub-block can be modeled as
\begin{equation}\label{eq:SIMO}
	\mathbb{Y}=
	\sqrt{\snr}\left(\mathsf{H}\cdot\mathsf{diag}\left\{\exp\left(j\bm{\Theta}\right)\right\}\cdot\mathbf{g} + \mathbf{d} \right)\trans{\mathbf{X}}+\mathbb{Z},
\end{equation}
where the joint channel input $\left(\mathbf{X},\bm{\Theta} \right)$ consists of two parts:
\begin{enumerate}
	\item the complex-valued random column vector $\mathbf{X}\in \mathbb{C}^{L}$ denotes the input signal in the sub-block from the transmit antenna, and the following trace constraint
	\begin{flalign}\label{eq:input_constraints}
		\mathbb{E}\left[{\hermi{\mathbf{X}}\mathbf{X}}\right]/L\le 1
	\end{flalign}
	is imposed on $\mathbf{X}$;
	\item the real-valued column vector $\bm{\Theta}=\trans{\left(\Theta_1,\cdots,\Theta_{n}\right)}\in [0,2\pi]^n$ denotes (possibly random) phase shifts at $n$ reflecting elements, and the unitary diagonal matrix $\mathsf{diag}\left\{\exp\left(j\bm{\Theta}\right)\right\}$ with the main diagonal entries $\exp\left(j\Theta_{1}\right)$, $\exp\left(j\Theta_{2}\right)$, $\cdots$, and $\exp\left(j\Theta_{n}\right)$
	denotes the reflecting coefficient matrix induced by the RIS;
\end{enumerate} 
where the deterministic matrix $\mathsf{H}\in \mathbb{C}^{\nr\times n}$, the column vector $\mathbf{g}\in \mathbb{C}^{n}$, and the column vector $\mathbf{d}\in \mathbb{C}^{\nr}$ denote constant channel gains from the RIS to the receive antennas, the transmitter to the RIS, and the transmitter directly to the receiver, respectively; where the random matrix $\mathbb{Z}=[Z_{ij}]\in \mathbb{C}^{\nr\times L}$ represents the channel noise with components being i.i.d. standard complex Gaussian random variables, i.e., $Z_{ij} \sim \mathcal{CN}(0,1)$; and where the positive scaling factor $\snr$ is used to represent the SNR of the considered system.

\subsection{Spatial Correlation}\label{sec:spatial_correlation}
As pointed out by \cite{Bjornson2021wcl}, modeling the wireless channel by independent Rayleigh fading may not capture the main characteristic of the RIS-aided wireless channels, especially spatial correlation caused by densely-packed phase-controllable elements.
	
With the assumption of isotropic scattering in the half-space in front of a planar RIS deployed on a two-dimensional rectangular grid with $n_X$ columns, $n_Y$ rows, and an inter-element spacing $d$, \cite[Sec.\Rmnum{3}-A]{Bjornson2021wcl} statistically models the channel gain $\trans{\myvec{h}}_k$ (row vector) between the $k$-th receive antenna and the RIS by
\begin{flalign}
	\myvec{h}_k \sim \mathcal{CN}(\bm{0}_{n},\mu_{{ \text{RR}}}^2\mathsf{R}),~\forall k\in[\nr],
\end{flalign}
where $n=n_{X}\cdot n_{Y}$ is the number of controllable elements, $\mu_{{ \text{RR}}}$ is the channel attenuation coefficient between the RIS and the receiver side, and the positive semi-definite matrix $\mathsf{R}=[r_{ij}]\in \mathbb{C}^{n\times n}$ can be determined as follows. 

For the $i$-th element and the $j$-th one respectively located at $(x_i,y_i,0)$ and  $(x_j,y_j,0)$, the corresponding spatial correlation coefficient $r_{ij}$ is given by
\begin{flalign}
	r_{ij} = \text{sinc} \bigg( \frac{2  \sqrt{(x_i-x_j)^2+(y_i-y_j)^2}}{\lambda} \bigg),~i,j\in[n],
\end{flalign}
where $\lambda$ is the wavelength of the propagating electromagnetic wave. 

In a similar way, we can model the channel gain $\myvec{g}$ between the single transmit antenna and the RIS by 
\begin{flalign}
	\mathbf{g} &\sim \mathcal{CN}(\mathbf{0}_{n}, \mu_{\text{TR}}^2 \mathsf{R}) ,
\end{flalign}
while, due to sufficient spacing at the receiver side, the channel gain  $\mathbf{d} $ of the line-of-sight link  follows 
\begin{flalign}
	\mathbf{d} &\sim \mathcal{CN}(\mathbf{0}_{\nr}, \mu_{\text{LOS}}^2 \mathsf{I}_{\nr}) ,
\end{flalign}
where $\mu_{{ \text{TR}}}$ denotes the channel attenuation coefficient between the transmit antenna and the RIS, $\mu_{{ \text{LOS}}}$ denotes that between the transmit antenna and the receiver side, and $\mathsf{I}_{\nr}$ denotes the $\nr \times \nr$ identity matrix.

\subsection{Equivalent Channel}
It has been shown in our previous work \cite{chen2024twc} that the original channel model~\eqref{eq:SIMO} can be simplified into the \textit{equivalent full-row-rank and purely-reflective model} as follows
	\begin{equation}\label{eq:equiv}
		\check{\mathbb{Y}}=\sqrt{\snr}\check{\mathsf{H}}\exp\left(j\check{\bm{\Theta}}\right)\trans{\check{\mathbf{X}}}+\check{\mathbb{Z}},
	\end{equation}
	where the random vector $\check{\mathbf{X}}$ is a phase-rotated channel input satisfying the original power constraint \eqref{eq:input_constraints}, the equivalent channel matrix $
	\check{\mathsf{H}}\triangleq\mathsf{diag}\left\{\rho_1,\cdots,\rho_{\tau}\right\}\cdot\hermi{\mathsf{V}_{1}} \in \mathbb{C}^{\tau \times \check{n}}$, the equivalent AWGN $\check{\mathbb{Z}}= [\check{Z}_{ij}]\in \mathbb{C}^{\tau \times L}$ with i.i.d. components $\check{Z}_{ij} \sim \mathcal{CN}\left( 0,1 \right)$ due to the rotation-invariant property of the AWGN, the integer $\tau$ denotes the rank of the equivalent channel matrix $\check{\mathsf{H}}$, and $\check{n}=n+1$ is the number of equivalent RIS elements (containing an extra virtual one accounting for the direct path). Due to the limited space, we refer the reader to \cite[Sec. \Rmnum{2}-B]{chen2024twc} for more details \review{on} computing the equivalent channel matrix $\check{\mathsf{H}}$. \review{Unless otherwise specified}, the remaining part is only concerned with the simplified model \eqref{eq:equiv}.

\section{Transceiver Architecture Based on Unitary Transformation}\label{sec:transceiver_architecture}

In this section, we first introduce a universal transceiver architecture for the multiple-antenna system, \review{which is based on the combination of unitary rotation of the channel matrix and successive interference cancellation.}

\subsection{Transceiver Architecture}
\review{To simplify the implementation}, we follow the convention, i.e., factoring the equivalent channel matrix $\check{\mathsf{H}}$ into a simple form via unitary rotation. To this end, we rewrite the equivalent channel matrix $\check{\mathsf{H}}$ as follows:
\begin{equation}
\check{\mathsf{H}} =\mathsf{B}\cdot \mathsf{C},
\end{equation}
where the \textit{rotation matrix} $\mathsf{B}$ is a $\tau \times \tau$ unitary matrix, the \textit{coefficient matrix} 
$\mathsf{C}=[c_{ij}] \in \mathbb{C}^{\tau \times \check{n}}$ is of row-echelon form, which will be specified later. From the consideration of regularization, we assume that the first nonzero element in each row of $\mathsf{C}$ is a positive number, whose index is denoted by $\bar{k}_i \ge i$ for each row index $i\in [\tau]$.

\subsubsection{Receiver}

The receiver first linearly transforms the channel output by the unitary matrix $ \mathsf{B}^{\dagger}$, and gets
\begin{flalign}
		\tilde{\mathbb{Y}} = \mathsf{B}^{\dagger}\check{\mathbb{Y}}= \sqrt{\snr}\cdot \mathsf{C} \cdot \exp\!\big(j\check{\boldsymbol{\Theta}}\big) \cdot \trans{\check{\mathbf{X}}} + \check{\mathbb{Z}},
\end{flalign}
which can be further regarded as $\tau$ vector subchannels as follows
\begin{equation}
	\trans{\tilde{\myvec{Y}}}_i=\sum_{k=\bar{k}_i}^{\check{n}}  \sqrt{\snr} c_{ik} \exp(j\check{\Theta}_k) \cdot \trans{\check{\mathbf{X}}} + \trans{\myvec{Z}}_i,\quad i \in [\tau].
\end{equation}
The receiver then successively decodes each subchannel from bottom to top. 

\subsubsection{Transmitter}
For simplification, we make the following two assumptions on the transmitter architecture:
\begin{itemize}
	\item[(A1)] The active vector signal $\check{\mathbf{X}}$ is independently and identically Gaussian distributed;
	\item[(A2)] The phase-shift elements appearing in the last subchannel are used to amplify the active signal $\check{\mathbf{X}}$ via beamforming, i.e.,  $\check{\Theta}_{k}=-\angle{c_{\tau k}}$ for $k\ge \tau$ and nonzero $c_{\tau k}$.
\end{itemize}
 
 \subsection{Performance Metric}
 In this paper, we measure the throughput performance of the considered symbiotic communication system by the achievable DoF (averaged over $L$ channel uses) of the above transceiver architecture, which is the sum of the achievable DoFs of $\tau$ individual subchannels, denoted by $\mathsf{DoF}_{i}$ for $i\in [\tau]$.

\review{Along the same lines as} \cite[Sec. \Rmnum{5}-A]{chen2024twc}, the $\tau$-th subchannel used for beamforming can be reduced to 
\begin{flalign}\label{eq:beamforming_subchannel}
		\trans{\tilde{\myvec{Y}}}_\tau= \sqrt{\snr}\Bigg(\sum_{k=\bar{k}_\tau}^{\check{n}} \left| c_{\tau k} \right| \Bigg)  \trans{\check{\mathbf{X}}} + \trans{\myvec{Z}}_\tau,
\end{flalign}
which has $\mathsf{DoF}_{\tau}=1$, and after the SIC, the $(\tau-1)$-th subchannel to the first one can be reduced to
\begin{flalign}\label{eq:subchannel}
	\tilde{{Y}}_{i}&=\sqrt{\snr}\left\|\check{\mathbf{x}}\right\|_2\left( \sum_{k=\bar{k}_i}^{\bar{k}_{i+1}-1}c_{ik}\exp(j \check{\Theta}_k
	) \right)+Z_{i}, 
\end{flalign}
which corresponds to $\mathsf{DoF}_{i}\in \left\{{1}/{L}, {2}/{L}\right\}$ for all $i \in [\tau-1]$; see Lemma~\ref{lem:achievable_magnitude} for more details. 
In the remaining part of this paper, we only consider the case with $L=1$ for simplicity \review{ and} the proposed scheme can be easily generalized to the case where $L\geq 2$.

%
%
%

\subsection{Existing Method}
In our previous work \cite{chen2024twc}, the proposed architecture uses the QR decomposition of $\check{\mathsf{H}}$, and hence, the induced matrix $\mathsf{C}$ is an upper triangular matrix with nonzero main diagonal elements, i.e., $\bar{k}_{i}=i$. Combined with SIC, the receiver can decode the phase-modulated information of the $i$-th RIS element from the $i$-th subchannel corresponding to the $i$-th row of the channel output $\check{\mathbb{Y}}$ from $i=\tau-1$ to $i=1$, i.e., in a bottom-to-top manner. \crh{A main drawback is that only a half DoF is achievable by each subchannel (i.e., $\mathsf{DoF}_{i}={1}/{2}$ for all $i\in [\tau-1]$), except for the $\tau$-th one which adopts the beamforming scheme via the remaining $\check{n}-\tau+1$ phase-shift elements.}

For this reason, the total achievable DoFs of the QR-SIC transceiver (with the Gaussian codebook at the transmitter side) are given by
\begin{flalign}
	\mathsf{DoF}_{\text{QR-SIC}}=1+\frac{\tau-1}{2},
\end{flalign}
which is strictly less than the maximum achievable DoF \cite{Hei2024TIT} 
\begin{flalign}
	\mathsf{DoF}_{\max}=\min \left\{ \tau,  \frac{\check{n}+1}{2}\right\}
\end{flalign}
except for the special case of $\tau=\check{n}$, i.e., spatially independent channels. This observation motivates our modulation scheme, proposed in the next section, to exploit the multiplexing potential in the spatial-correlated regime.

\section{Grouped Annulus-Modulated Transceiver}\label{sec:GAM}

\review{Before formally introducing our proposed transceiver architecture}, we present a heuristic lemma that motivates our design.

\subsection{Constellation Geometry for Phase Modulation}

The following lemma characterizes the constellation geometry formed by the linear combination of several phase-modulated signals.

\begin{lemma}
	\label{lem:achievable_magnitude}
	Given a complex column vector $\mathbf{p} \in \mathbb{C}^m$ ($m\ge 2$), denote the optimal value of the following problem
	\begin{subequations}
		\begin{align}
			&\min && \bigg| \sum_{i=1}^m \epsilon_i |p_i|  \bigg|\\
			&\text{s.t.}&&\epsilon_i\in \left\{-1,1\right\},~\forall i \in[m],
		\end{align}
	\end{subequations}
	by $\const{r}_{\min}$ and $\sum_{i=1}^m |p_i| $ by $\const{r}_{\max}$. Then we have
	\begin{flalign}\label{eq:modulus_n2}
		\mathcal{P} &= \{\trans{\myvec{p}} \exp(j\bm{\theta}) \mid \bm{\theta} \in [0, 2\pi]^m \} \nonumber \\
		&= \{\rho \exp(j\phi) \mid \rho \in [\const{r}_{\min},\const{r}_{\max}] \text{ and } \phi \in [0, 2\pi] \}.
	\end{flalign}
\end{lemma}

\begin{IEEEproof}
	Note that the set $\mathcal{P}$ is rotation invariant by showing that, for any $\phi \in [0,2\pi]$, $\rho \exp(j\phi) \in \mathcal{P}$ if $\rho \in \mathcal{P}$. Thus, we only need to focus on the modulus set $\mathcal{A}$. Without loss of generality, we can let the complex vector $\myvec{p}$ be nonnegative real vector $\trans{(|p_1|, |p_2|, \cdots, |p_m|)}$.
	
	We first consider the simplest case, where  $m=2$ and  the modulus can be therefore expressed as
	\begin{flalign}
		&||p_1|\exp(j\theta_1) + |p_2|\exp(j\theta_2)| \nonumber \\
		&~= \sqrt{(|p_1|\cos\theta_1 + |p_2|\cos\theta_2)^2 + (|p_1|\sin\theta_1 + |p_2|\sin\theta_2)^2} \\
		&~= \sqrt{|p_1|^2 + |p_2|^2 + 2|p_1||p_2|\cos(\theta_1 - \theta_2)} \\
		&~\in [|(|p_1| - |p_2|)|, |p_1| + |p_2|],
	\end{flalign}
	i.e., confirming Eq.~\eqref{eq:modulus_n2}. Then we can conclude the lemma by mathematical induction.
\end{IEEEproof}


It is straightforward to see that there is a phase transition between the case of $m=1$ and that of $m\ge 2$. In the former case, the constellation geometry is merely a one-dimensional complex circle, while in the latter case, the constellation geometry is a two-dimensional annulus. See Fig.~\ref{Fig:constellation_geometry} for illustration. 

\crh{
\subsection{Main Idea of Grouped Annulus-Modulated Transceiver}

As seen from Lemma~\ref{lem:achievable_magnitude},} an encouraging fact is that if we group at least two elements in the  subchannel \eqref{eq:subchannel}, its achievable DoF can be doubled as compared with that of the QR-SIC transceiver proposed in \cite{chen2024twc}. 
To this end, we shall propose an unconventional matrix decomposition method that unitarily transforms the equivalent channel matrix $\check{\mathsf{H}}$ into a special row-echelon form, where each row has two more non-zero entries than the row below it. 
For more clarity, please see Fig.~\ref{Fig:row-echelon form}. It can be seen from Eq.~\eqref{eq:subchannel} and Fig.~\ref{fig:row_echelon_QR} that the QR decomposition leads to a strictly upper-triangular row-echelon form, and hence, yields $\tau-1$ phase-modulated subchannels with circular constellation geometry, while the proposed decomposition method leads to a specific row-echelon form with two dominant coefficients per row, thereby forming the subchannels (after SIC) as follows
\begin{flalign}
	\tilde{Y}_i &=\sqrt{\snr} \| \check {\mathbf{x}} \|_2
	\big(c_{i,2i-1} \exp(j \check{\Theta}_{2i-1})  \nonumber\\
	&\quad\quad\quad~~+  c_{i,2i} \exp(j \check{\Theta}_{2i})
	\big) +Z_i,~i\in[\tau-1],
	\label{eq:effective_subchannel}
\end{flalign}
with an annular constellation geometry with respect to the equivalent transmitted signal $c_{i,2i-1} \exp(j \check{\Theta}_{2i-1})+c_{i,2i} \exp(j \check{\Theta}_{2i})$. 

\crh{
In this paper, we refer to the transceiver architecture using the special coefficient matrix $\mat{C}$ (see Fig.~\ref{fig:row_echelon_proposed}) and SIC as \textit{the grouped annulus-modulated (GAM) transceiver}. }

\begin{figure}[t]
	\centering
	\subfigure[$m=1$]{\label{fig:geometry_m1}
		\begin{minipage}{0.40\textwidth}
			\includegraphics[width=\linewidth]{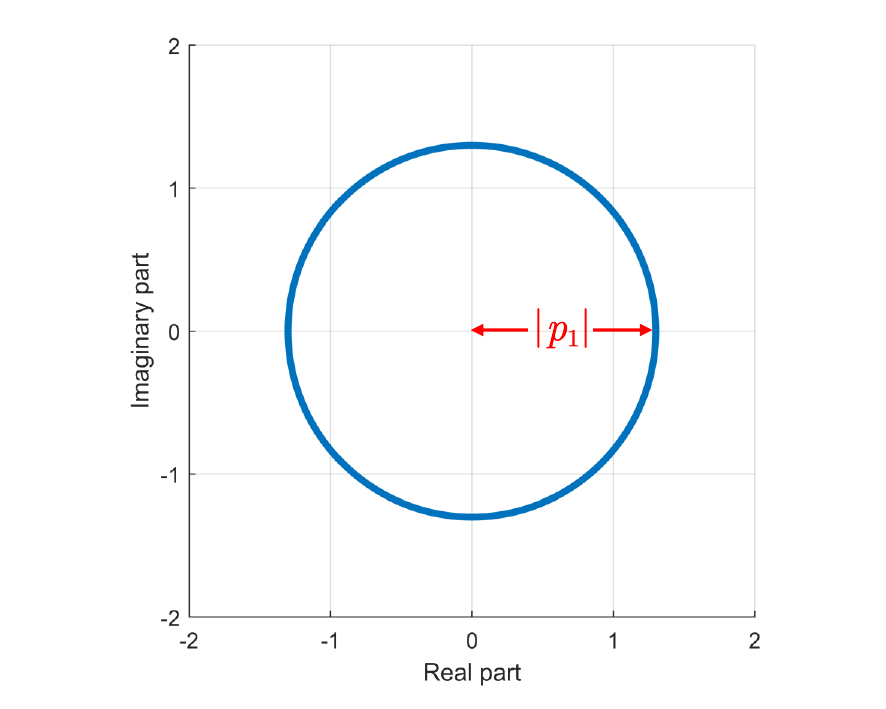}
		\end{minipage}
	}
	
	\subfigure[$m\ge2$]{\label{fig:geometry_m_ge_2}
		\begin{minipage}{0.40\textwidth}
			\includegraphics[width=\linewidth]{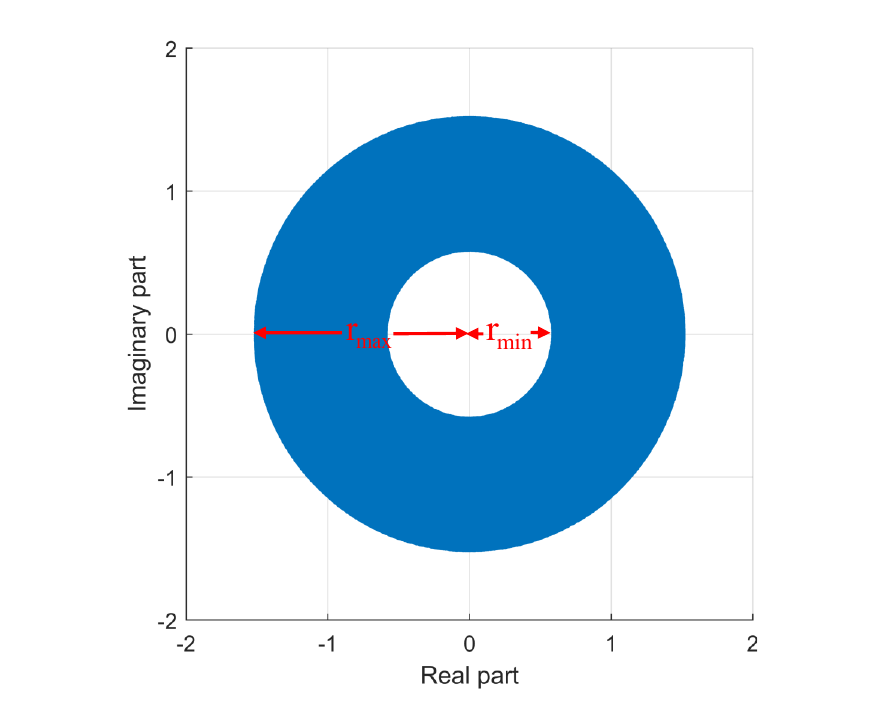}
		\end{minipage}
	}
	\caption{Constellation geometries of the linear combination of $m$ phase-modulated signals.} \label{Fig:constellation_geometry}
\end{figure}

\begin{figure}[t]
	\centering
	\subfigure[The coefficient matrix $\mat{R}$ via QR Decomposition~\cite{chen2024twc}]{
		\begin{minipage}{0.45\textwidth}
			\includegraphics[width=\linewidth]{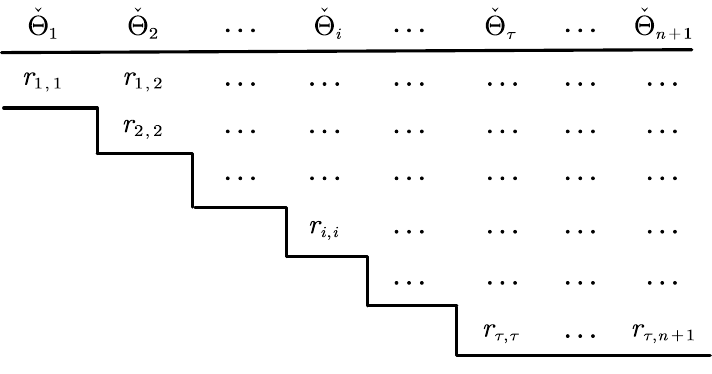}\label{fig:row_echelon_QR}
		\end{minipage}
	}
	
	\subfigure[The coefficient matrix $\mat{C}$ expected in this paper]{
		\begin{minipage}{0.45\textwidth}
			\includegraphics[width=\linewidth]{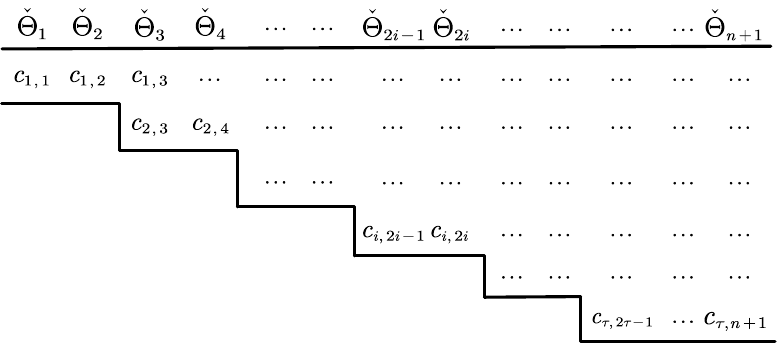}\label{fig:row_echelon_proposed}
		\end{minipage}
	}
	\caption{Two different types of coefficient matrices in	row-echelon forms.} \label{Fig:row-echelon form}
\end{figure}

\subsection{Least-Residual Matrix Decomposition Algorithm}
This subsection is devoted to the matrix decomposition algorithm for the proposed transceiver. Based on the basic knowledge of linear algebra, a full-row-rank $\tau\times \check{n}$ matrix should satisfy certain correlation conditions if it can be unitarily transformed into the special row-echelon form in Fig.~\ref{fig:row_echelon_proposed}. For this reason, we will first formulate the matrix decomposition problem of interest, where the residual error (compared with the ideal row-echelon form in Fig.~\ref{fig:row_echelon_proposed}) is taken into account.

\subsubsection{Problem Formulation}
We define the residual error as the sum of squares of the entries hoped to be zero. Then the least-residual unitary matrix $\mathsf{B}^{\star}$ is the solution to the following optimization problem:
	\begin{subequations}\label{prob:main}
	\begin{align}
		&\min_{\mathsf{B},\,\mathsf{P}} && \sum_{i=1}^{\tau}\sum_{j=1}^{2(i-1)} \Big( \trans{\myvec{e}}_i
		\big(
		\mathsf{B}^{\dagger} \check{\mathsf{H}}
		\mathsf{P}
		\big)
		\myvec{e}_j
		\Big)^2
		\\
		&\text{ s.t.}&&~~~~~\mathsf{B}\in {U}(\tau),\,\mathsf{P}\in P(\check{n})
	\end{align}
\end{subequations}
where $U(\tau)$ denotes the unitary group consisting of all $\tau\times \tau$ unitary matrices, and $\mathcal{P}(\check{n})$ denotes the permutation group of all $\check{n} \times \check{n}$ permutation matrices. Here we introduce the permutation matrix $\mathsf{P}$ to account for the freedom of arbitrarily sorting $\check{n}$ equivalent RIS elements.

Note that the problem~\eqref{prob:main} constitutes a hybrid combinatorial nonconvex  optimization problem, for which a closed-form solution seems to be intractable. Nevertheless, in this paper we propose a low-complexity algorithm, termed \textit{the combinatorially pairing (CP) algorithm}, which can obtain a near-ideal coefficient matrix especially for the case with large $\check{n}$. Before presenting the proposed algorithm, we first give a heuristic one that serves as a benchmark scheme. 

\subsubsection{Heuristic Algorithms}\label{sec:Heuristic}
Since there is no existing method that can be referred to, we set the following two heuristic methods as the benchmark. 
\begin{itemize}
	\item Random Rotation. We uniformly generate $10000$ random rotation matrices $\mat{B}$ from the unitary group $U(\tau)$ with respect to the Haar measure, and then choose the best one leading to  least residual error of its corresponding coefficient matrix.
	\item Gram-Schmidt Orthogonalization. In the $i$-th step ($i\in[\tau]$), we normalize the first vector among the remaining $(\check{n}-2(i-1))$ column vectors as the required unit basis vector in the $i$-th column of the factor matrix $\mathsf{B}$,  and then we update the remaining vectors via Gram-Schmidt \review{orthogonalization process}.
\end{itemize}

\subsubsection{Combinatorially Pairing Algorithm}
In the subsequent section, we present the CP algorithm that iteratively computes a least-residual unit-norm basis vector among all pairs of the remaining column vectors, instead of directly choosing one of them and then using Gram-Schmidt orthogonalization.

Given a set of input vectors $\{\mathbf{a}_1, \mathbf{a}_2, \dots, \mathbf{a}_n\}$, the core idea of the CP algorithm is  to find a unit vector $\mathbf{b}$ such that the sum of the two smallest 
projection errors of $\{\mathbf{a}_i\}$ onto $\mathbf{b}$ is minimized. 
Meanwhile, the corresponding error values for each pair are recorded 
for subsequent basis construction.

Before a formal presentation of the proposed algorithm, we first give the following lemma.

\begin{lemma}
	\label{lem:complex matrix}
	Given two nonzero $\tau$-dimensional column vectors $\myvec{a}_1,\myvec{a}_2\in\mathbb C^{\tau}$, the optimal unit-norm vector
	$\mathbf b^\star$ that minimizes the total projection
	error $\left\|  \myvec{a}_1 - (\myvec{b}^{\dagger}\myvec{a}_1)\myvec{b} \right\|_2^2
	+\left\|  \myvec{a}_2 - (\myvec{b}^{\dagger}\myvec{a}_2)\myvec{b} \right\|_2^2
	$ is given by the dominant eigenvector of the Gram matrix
	$\mathsf{A} \mathsf{A}^{\dagger}$, where $\mathsf{A}=[\myvec{a}_1,\myvec{a}_2]\in\mathbb{C}^{\tau\times 2}$.
\end{lemma}

\begin{proof}
	Denote the inner products $c_i=\myvec{b}^{\dagger}\myvec{a}_i$ for $i\in \left\{ 1,2\right\}$. Clearly, the matrix $\myvec{b}^{\star}\cdot [c_1,c_2]$ must be the best rank-one approximation of the matrix $\mathsf{A}$. Thus, the lemma can be easily concluded by using the Eckart-Young Theorem \cite{eckart1936approximation}.
\end{proof}

Based on the above lemma, the CP algorithm computes  residual errors for all pairs of two column vectors and selects the basis vector corresponding the least-residual pair. See the next subsection for detailed procedures.

\subsubsection{Implementation Procedures}
The algorithm executes in a pair-by-pair manner as follows:

\begin{itemize}
	\item[] \textbf{Input:} The equivalent channel matrix $\check{\mathsf{H}} \in \mathbb{C}^{\tau \times \check{n}}$.
	\item \textit{Step 1: Initialization.}
	Set the loop index variable $t=1$, the residual matrix as
	$\mathsf{A}^{(1)}=\check{\mathsf{H}}$, and the permutation matrix as $\mathsf{P}^{(1)}=\mathsf{I}_{\check{n}}$.
	
	\item \textit{Step 2: Combinatorial Pairing.} Enumerate all candidate unordered pairs $(\myvec{a}_{i_{t}},\myvec{a}_{j_{t}})$ of column vectors from the residual matrix $\mathsf{A}^{(t)}$, and compute the corresponding unit-norm basis vectors $b^{\star}_{i_t j_t}$ with total projection errors $e_{i_t j_t}$ by Lemma~\ref{lem:complex matrix}.
	
	\item \textit{Step 3: Basis and Permutation Update.} Search the minimum projection error $e_{i^{\ast}_{t}j^{\ast}_{t}}$ and add the corresponding basis column vector $b^{\star}_{i^{\ast}_{t}j^{\ast}_{t}}$ into the basis matrix $\mathsf{B}^{(t)}$. Then update the permutation matrix $\mathsf{P}^{(t)}$ such that $(i^{\ast}_{t},j^{\ast}_{t})$ is permutated to $(2t-1,2t)$.
	
	\item \textit{Step 4: Residual Matrix Update.} Remove the $i^{\ast}_{t}$-th and the $j^{\ast}_{t}$-th columns of the residual matrix $\mathsf{A}^{(t)}$, and update it as $\mathsf{A}^{(t+1)} \leftarrow \mathsf{A}^{(t)}-\myvec{b}\big(\big(\myvec{b}^{\star}_{i^{\ast}_{t}j^{\ast}_{t}}\big)^{\dagger}\mathsf{A}^{(t)}\big) $.
	
	\item \textit{Step 5: Iteration.}
	Set $t \leftarrow t+1$, and repeat Steps 2--4 for the loop index variable $t$ from $2$ to $\tau$.
	
	\item[] \textbf{Output:} Unitary matrix $\mathsf{B}=\mathsf{B}^{(\tau)}$, permutation matrix $\mathsf{P}=\mathsf{P}^{(\tau)}$, and coefficient matrix $\mathsf{C}=\mathsf{B}^{\dagger} \check{\mathsf{H}}\mathsf{P}$.
\end{itemize}

\subsubsection{Complexity Evaluation }

Clearly, the overall computation complexity of the CP algorithm is $O(\check{n}^2)$, which primarily stems from repeatedly invoking Lemma~\ref{lem:complex matrix} $\tau$  times.\footnote{The exact number of invocations is given as $\sum_{t=0}^{\tau-1} \binom{\check{n}-t}{2}= \frac{\tau}{2} \check{n}^2 - \frac{\tau^2}{2} \check{n} + \frac{\tau^3 - \tau}{6}$.}

\subsubsection{Numerical Example}
For the readers' better understanding of the CP algorithm as well as its performance, here we numerically compute coefficient matrices for several spatially correlated channels, respectively via the CP algorithm and the QR decomposition. 
According to the prescribed channel model in Sec.~\ref{sec:spatial_correlation} and the channel reduction method introduced in  \cite[Sec.~\Rmnum{2}-B]{chen2024twc}, we generate two realizations $\check{\mat{H}}_1$ and $\check{\mat{H}}_2$ of the equivalent channel matrix for three-receiver channels with $n=6$ and $n=1024$ RIS elements, where channel attenuation coefficients are set to $\mu_{\text{LOS}}=-25$~dB and $\mu_{\text{RR}}=\mu_{ \text{TR}}=-5$~dB, respectively. The corresponding realizations of the equivalent matrices $\check{\mathsf{H}}_1 \in \mathbb{C}^{3\times 7}$ and $\check{\mathsf{H}}_2 \in \mathbb{C}^{3\times 1025}$ are shown at the top of the next page.

For each equivalent channel matrix example $\check{\mathsf{H}}$, the QR decomposition yields a strictly upper-triangular coefficient matrix $\mathsf{R}$, \review{while the coefficient matrix $\mathsf{C}$ resulting from the proposed CP algorithm approximately exhibits an expected stepped structure as illustrated in Fig.~\ref{fig:row_echelon_proposed}}. Moreover, as the number of column vectors of the equivalent channel matrix (i.e., the number of equivalent RIS elements) increases, there is more space for the CP algorithm to substantially suppress projection residuals, which further confirms the effectiveness of our scheme.

\crh{
\subsection{DoF Analysis}
With assumptions of ideal SIC and negligible residual error, the total achievable DoFs of the GAM transceiver can be determined case-by-case:
	\begin{itemize}
		\item \textit{Case 1: $\review{\check{n}}\ge 2\tau-1$.} All $\tau-1$ phase-modulated subchannels \eqref{eq:effective_subchannel} of the GAM transceiver have $\mathsf{DoF}_{i}=1$ for all $i\in [\tau-1]$, due to their two-dimensional annular constellation geometries illustrated by Fig~\ref{fig:geometry_m_ge_2}. Thus, the total achievable DoFs of the GAM transceiver is given as
		\begin{flalign}\label{eq:DoF_GAM}
			\mathsf{DoF}_{\text{GAM}}=\tau=	\mathsf{DoF}_{\max}.
		\end{flalign}
		\item \textit{Case 2: $\tau \le \review{\check{n}}\le 2\tau-2$.} With a slight change, the GAM transceiver still can attain full DoFs. Let the first $\check{n}-\tau$ phase-modulated subchannels use the annular constellation geometry, and the others use individual phase-only modulation (corresponding to the circle constellation geometry illustrated by Fig~\ref{fig:geometry_m1}). Then the total achievable DoFs can be determined as
		\begin{flalign}
			\mathsf{DoF}_{\text{GAM}}&=(\review{\check{n}}-\tau)+\frac{1}{2}(\tau-1-(\review{\check{n}}-\tau))+1\\
			&=\frac{\review{\check{n}}+1}{2} \\
			&=\mathsf{DoF}_{\max}.
		\end{flalign} 
	\end{itemize}
	Thus, if the residual error induced by the CP algorithm is sufficiently small, the GAM transceiver has the potential of achieving full DoFs for either spatially correlated or independent channel.\footnote{We would like to remind the reader that the performance of the proposed CP algorithm significantly improves when the number of RIS element becomes large, which corresponds to strong spatial correlation and aligns with practical wireless systems. }
}

\begin{figure*}[!t]
	\[
	\scalebox{0.81}{$
		\setlength{\arraycolsep}{6pt}
		\check{\mathsf{H}}_1 =
		\begin{bmatrix}
			-2.0 + 1.6i &
			-3.6 + 0.1i &
			-3.4 - 2.8i &
			-1.3 + 2.1i &
			-2.6 + 0.7i &
			-2.0 - 2.0i &
			-0.2 - 0.1i \\
			
			-0.4 + 0.6i &
			0.1 + 0.0i &
			0.9 + 0.1i &
			-0.2 + 1.0i &
			-0.1 + 0.3i &
			0.3 + 0.2i &
			-0.6 - 0.0i \\
			
			-0.1 + 0.1i &
			0.1i &
			0.0i &
			0.2 + 0.2i &
			0.1 + 0.0i &
			0.2 - 0.1i &
			0.4 - 0.4i
		\end{bmatrix}
		$}
	\]
	
	\[
	\scalebox{0.81}{$
		\setlength{\arraycolsep}{6pt}
		\mathsf{R}_1 =
		\begin{bmatrix}
			2.7  &
			2.7 + 2.1i &
			0.7 + 3.9i &
			2.5 - 0.9i &
			2.5 + 1.0i &
			0.2 + 2.6i &
			0.2 + 0.4i \\
			
			0 &
			1.1 &
			1.7 + 1.1i &
			-0.1 - 0.2i &
			0.6+ 0.03i &
			0.8 + 0.8i &
			-0.4 - 0.2i \\
			
			0 &
			0 &
			0.2  &
			0.3i &
			0.5 + 0.2i &
			0.1 + 0.2i &
			0.4 + 0.5i
		\end{bmatrix}
		$}
	\]
	
	\[
	\scalebox{0.81}{$
		\setlength{\arraycolsep}{6pt}
		\mathsf{C}_1 =
		 \begin{bmatrix} 
			2.6 &
			2.5 - 0.9i & 
			2.7 + 2.0i & 
			0.8 + 3.8i & 
			2.5 + 1.0i & 
			0.2 + 2.6i & 
			0.2 + 0.4i \\
			
			0.07 - 0.1i &
			-0.02 + 0.1i &
			1.3 &
			1.7 + 1.3i & 
			0.7 + 0.01i & 
			0.8 + 0.9i & 
			-0.4 - 0.1i \\
			
			-0.08 - 0.08i &
			0.1 + 0.05i & 
			0.03i & 
			0.01 - 0.01i & 
			0.1 & 
			0.2 - 0.3i &
			0.6 - 0.3i 
		\end{bmatrix}
		$}
	\]
	
	\[
	\scalebox{0.81}{$
		\setlength{\arraycolsep}{6pt}
		\check{\mathsf{H}}_2 =
		\begin{bmatrix}
			0.2 + 0.4i &
			0.8 + 0.07i &
			0.4 - 0.6i &
			-0.1 - 1.6i &
			0.07 - 2.2i &
			0.8 - 2.0i &
			1.2 - 1.6i &
			\cdots \\
			
			-0.1 + 0.09i &
			-0.4 + 0.7i &
			-1.0 + 1.2i &
			-0.9 + 1.0i &
			-0.6 + 0.5i &
			-0.07 + 0.2i &
			0.3 + 0.3i &
			\cdots \\
			
			0.2 - 0.1i &
			0.1 - 0.7i &
			-0.09 - 0.9i &
			-0.5 - 0.9i &
			-0.7 - 0.8i &
			-0.6 - 0.8i &
			-0.6 - 0.7i &
			\cdots
		\end{bmatrix}
		$}
	\]

	\[
	\scalebox{0.81}{$
		\setlength{\arraycolsep}{6pt}
		\mathsf{R}_2 =
		\begin{bmatrix}
			0.6 &
			0.8 - 0.8i &
			0.4 - 1.0i &
			-0.8 - 1.1i &
			-1.4 - 1.5i &
			-1.1 - 2.0i &
			-0.7 - 2.1i &
			\cdots \\
			
			0 &
			0.5 &
			1.5 - 0.03i &
			1.9 + 0.2i &
			1.4 + 0.5i &
			0.7 + 0.5i &
			0.2 + 0.3i &
			\cdots \\
			
			0 &
			0 &
			0.07  &
			0.2 + 0.05i &
			0.2 + 0.2i &
			0.1 + 0.2i &
			0.02 + 0.09i &
			\cdots
		\end{bmatrix}
		$}
	\]
	
	\[
	\scalebox{0.81}{$
		\setlength{\arraycolsep}{4pt}
		\mathsf{C}_2 =
		\begin{bmatrix}
			1.0  &
			0.2  &
			-0.3 + 0.1i &
			-0.5 + 0.2i &
			0.03i &
			0.02i &
			\cdots \\
			
			9.5\times10^{-6} + 9.0\times10^{-6}i &
			-7.5\times10^{-5} - 4.9\times10^{-5}i &
			0.4  &
			0.3 - 0.1i &
			0.03 - 0.01i &
			-0.3 + 0.1i &
			\cdots \\
			
			1.2\times10^{-5} + 2.1\times10^{-5}i &
			-1.1\times10^{-3} - 1.2\times10^{-4}i &
			-2.0\times10^{-4} + 6.5\times10^{-5}i &
			2.13\times10^{-4} - 2.1\times10^{-4}i &
			0.1 &
			-0.03 - 0.02i &
			\cdots
		\end{bmatrix}
		$}
	\]

	\vspace{0.5em}
	\noindent\rule{\textwidth}{0.4pt}
\end{figure*}

\section{Annular Constellation With Hexagonal Lattice Points}\label{Sec.Annular Constellation}
In this section, we further consider the constellation design problem for the proposed GAM transceiver. For the $i$-th subchannel~\eqref{eq:effective_subchannel}, we define its feasible set of the equivalent transmitted signal as
\begin{flalign}
	\mathcal{A}_i \triangleq
	\big\{&
		c_{i,2i-1} \exp(j \check{\Theta}_{2i-1}) 
		+  c_{i,2i} \exp(j \check{\Theta}_{2i}) : \nonumber\\
&~\forall\, \check{\Theta}_{2i-1},\check{\Theta}_{2i}
\in [0,2\pi]	\big\}, ~\forall i\in [\tau-1].
\end{flalign}
Without loss of generality, here we assume that the $n$ real RIS elements and an extra one virtual RIS element formed by the direct path have been sorted according to the permutation matrix $\mathsf{P}$ \crh{outputted} by the CP algorithm.

It follows from Lemma~\ref{lem:achievable_magnitude} that the set $	\mathcal{A}_i$ is an \review{annulus with the inner radius $\const{r}_{ \text{in}}^{(i)}=\left|(|c_{i,2i-1}| - |c_{i,2i}|)\right|$ and the outer radius $\const{r}_{ \text{out}}^{(i)}=|c_{i,2i-1}| + |c_{i,2i}|$.}

The goal of constellation design for $\tau-1$ phase-modulated subchannels is to pack a certain number of constellation points in the \review{annuli} $\mathcal{A}_1$, $\ldots$, and $\mathcal{A}_{\tau-1}$ as densely as possible. Since the hexagonal lattice is the densest lattice packing in the two-dimensional Euclidean space, we let the joint constellation for the two phase-modulated signals $\check{\Theta}_{2i-1}$ and $\check{\Theta}_{2i}$ be
\begin{flalign}
	\mathcal{S}_i = \mathcal{A}_i \cap (\eta_i \mathcal{H})
\end{flalign}
for $i\in [\tau-1]$, where the hexagonal lattice $\mathcal{H}=\left\{
z_1+z_2 \exp(j\pi/3): z_1,z_2 \in \mathbb{Z}
\right\}$ and the scaling factor $\eta_i$ is adjusted according to the required data rate or the expected level of error tolerance.

It is well-known that the error performance of the multidimensional constellation over the AWGN channel is dominated by its minimum Euclidean distance (MED) \cite{Forney1989Multidimensional1}, especially in the high-SNR regime. For the constructed constellation $\mathcal{S}_i$, its MED is the product of the scaling factor $\eta_i$ and the MED of the hexagonal lattice $\mathcal{H}$, i.e.,
\begin{flalign}
	\mathsf{MED}(\mathcal{S}_i)=\eta_i.
\end{flalign}
Then the uncoded symbol error rate (SER) for the $i$-th phase-modulated subchannel can be approximated by the union upper bound \cite[Eq. {(8.12-1)}]{Proakis2008}
\begin{flalign} \label{eq:APSK}
		P_{e}^{(i)}\approx 6  Q
		\left( \frac{\eta_i}{\sqrt{2}} \right),
\end{flalign}
where the coefficient $6$ follows from the kissing number of the hexagonal lattice, and $Q \left( u \right)\triangleq \dfrac1{\sqrt{2\pi}}\int_u^\infty e^{-y^2/2}dy$ is the Gaussian Q-function. Note that this approximation becomes invalid when the residual error caused by the CP algorithm and SIC cannot be neglected.


\crh{
\subsection{Fast Enumeration of Constellation Points}
For the purpose of constellation mapping, we should enumerate all points in the annular constellation $\mathcal{S}$ with inner radius $\const{r}_{\text{in}}$, outer radius $\const{r}_{\text{out}}$, and the factor $\eta$ used for scaling the hexagonal lattice. The problem is equivalent to enumerating all points $\mathbf{a}$ of the hexagonal lattice $\mathcal{H}$ satisfying the norm constraint
\begin{flalign}\label{eq:norm_constraint}
\left(
\frac{\const{r}_{\text{in}}}{\eta}
\right)^2
	\le 	 \left|\mathbf{a} \right|^2
	\le 
\left(\frac{\const{r}_{\text{out}}}{\eta}
\right)^2.
\end{flalign} 

We first introduce the following definition of the \textit{theta series} of the hexagonal lattice \cite{conwaysloane99_1}
	\begin{flalign}
		\Theta_{\text{hex}}(q)&\triangleq \sum_{\myvec{a}\in \mathcal{H}}
		q^{\left| \myvec{a} \right|^2} \\
		&=
		1+6q+6q^3+6q^4+12q^7+\cdots,
	\end{flalign}
	for which we denote the coefficient of the term $q^k$ by $\const{N}_{\text{hex}}(k)$ for all $k\in \mathbb{N}$. Clearly, the quantity $\const{N}_{\text{hex}}(k)$ also represents the number of hexagonal lattice points of the modulus $\sqrt{k}$ and can be specified by
	\begin{flalign}
		\const{N}_{\text{hex}}(k)=6\bigg(
		\prod_{p \vert k} \chi(p)
		\bigg)
	\end{flalign}
	for any $k\ge 2$, where 
	\begin{flalign}\label{eq:negative_moment}
		\chi(p) =\begin{cases}
			1&,~p=3,\\
			\nu_p(k)+1&,~p=1~ (\text{mod } 3),\\
			0&,~p=2~ (\text{mod }  3) \text{ and odd } \nu_p(k),\\
			1&,~p=2~ (\text{mod }  3) \text{ and even } \nu_p(k),
		\end{cases}
	\end{flalign} 
	and $p$ is the distinct prime factor of $k$ with multiplicity $\nu_p(k)$.
	
	Hence, we immediately compute the size of the annular constellation $\mathcal{S}$ as
	\begin{flalign}
		\left|  
		\mathcal{S}
		\right|= 
			\const{N}_{\text{hex}}\left(
		\bigg\lfloor 
		\left(
		\frac{\const{r}_{\text{out}}}{\eta}
		\right)^2
		\bigg\rfloor
		\right)
		-
	\const{N}_{\text{hex}}\left(
	\bigg\lceil 
	\left(
	\frac{\const{r}_{\text{in}}}{\eta}
	\right)^2
	\bigg\rceil
	\right).
	\end{flalign}
	Then all hexagonal lattice points $z_1+z_2\exp(j\pi/3)$ satisfying the norm constraint~\eqref{eq:norm_constraint} can be characterized by two simple steps
	
	\begin{itemize}
		\item \textit{Step 1.}
		\review{
		For each integer $k\in [	\lfloor 
		\left(
		{\const{r}_{\text{in}}}/{\eta}
		\right)^2
		\rfloor,\lceil 
		\left(
		{\const{r}_{\text{out}}}/{\eta}
		\right)^2
		\rceil]$,} search for valid $k$ satisfying  $\const{N}_{\text{hex}}(k)\neq 0$.
	
		\item \textit{Step 2.} Numerically solve the Diophantine equation $z_1^2+z_2^2+z_1z_2 =k$.
	\end{itemize}

\subsection{Fast Decomposition of Constellation Points}
}
Given a targeted constellation point $s=\rho \exp(j\phi)\in \mathcal{S}_i$ for the $i$-th phase-modulated subchannel \eqref{eq:effective_subchannel}, it is required for signaling purposes to decompose $\rho \exp(j\phi)$ into two deterministic phase values $\check{\theta}_{2i-1}$ and $\check{\theta}_{2i}$ satisfying
\begin{flalign}\label{eq:decomposition}
	\rho\exp(j\phi)=c_{i,2i-1}\exp(j \check{\theta}_{2i-1}) + c_{i,2i}\exp(j \check{\theta}_{2i}).
\end{flalign}  
A simple decomposition method is given as follows.

Define $\varphi \triangleq \angle c_{i,2i} - \angle c_{i,2i-1} $. Then, by the law of cosines, we have the phase difference
\begin{flalign}
\Delta&\triangleq	\left|
	\check{\theta}_{2i}+\varphi -\check{\theta}_{2i-1} 
	\right|  \nonumber\\
	&=
	\arccos \left(\frac{\rho^2 - |c_{i,2i-1}|^2 - |c_{i,2i}|^2}{2|c_{i,2i-1}|\cdot|c_{i,2i}|}\right)
	  \in [0,\pi].
\end{flalign}
Without loss of generality, we can let $	\check{\theta}_{2i}= \Delta+\check{\theta}_{2i-1} -\varphi$ with $\check{\theta}_{2i-1}\in [0,2\pi]$. Substituting it into~\eqref{eq:decomposition}, we have the solution
\begin{flalign}
	\check{\theta}_{2i-1}=\angle
	\left(
	\frac{\rho \exp(j\phi)}{c_{i,2i-1}+c_{i,2i}\exp(j(\Delta-\varphi))}
	\right).
\end{flalign}

\section{Simulation Results}\label{Sec.Simulation result}

In this section, we present numerical results to evaluate the performance of the CP algorithm and the annular constellation proposed in our paper. We consider an RIS-assisted symbiotic communication system with a single-antenna transmitter and a four-antenna receiver (corresponding to $\tau=4$ in probability), where the planar $32\times 32$ RIS is deployed on a two-dimensional equally-spaced rectangular grid with inter-element spacing $d={\lambda}/{2}$, ${\lambda}/{4}$, or ${\lambda}/{8}$. The channel attenuation coefficients are set to $\mu_{\text{LOS}}=-60$ dB and $\mu_{\text{RR}}=\mu_{\text{TR}}=-5$ dB, \review{which correspond to the scenario with a weak line-of-sight link.}

\subsection{Residual Error Performance of Matrix Decomposition Algorithms}

In this subsection, we investigate the residual error performance of the CP algorithm, with the two heuristic algorithms in Sec.~\ref{sec:Heuristic} serving as the benchmarks. For a coefficient matrix $\mat{C}$ obtained by some matrix decomposition algorithm, we consider its \textit{relative residual error} (RRE) defined as follows
\begin{flalign}
	\text{RRE}=\frac{\sum_{i=2}^{\tau}
	\sum_{j=1}^{2(i-1)} c_{ij}^2}{\left\| \check{\mat{H}} \right\|_{F}^2}
	=\frac{\sum_{i=2}^{\tau}
		\sum_{j=1}^{2(i-1)} c_{ij}^2}{\left\| \check{\mat{C}} \right\|_{F}^2},
\end{flalign}
where $\left\| \cdot \right\|_{F}$ denotes the Frobenius norm.

\begin{figure}[!htp]
	\centering
	\includegraphics[width=0.95\columnwidth]{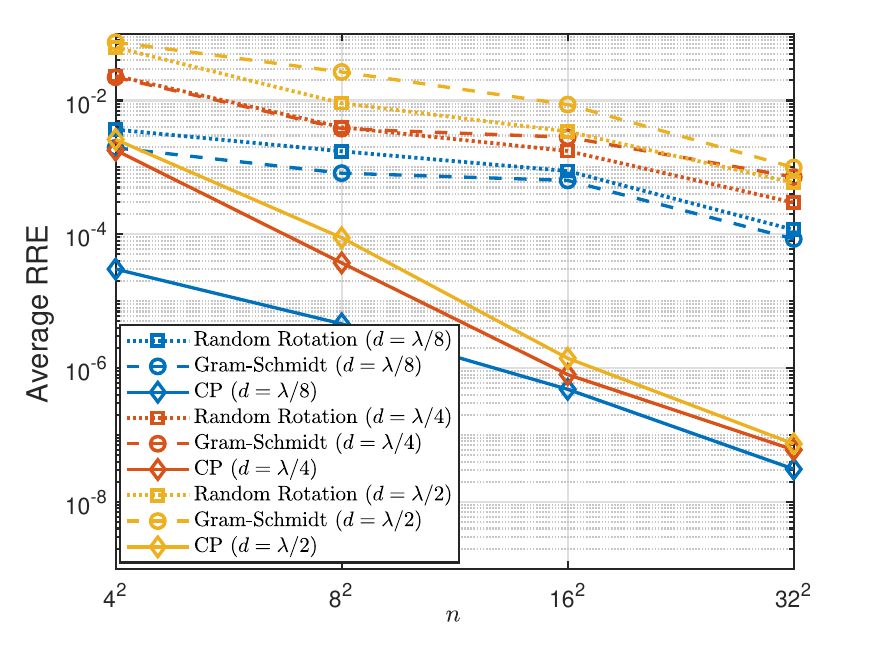}
	\caption{\crh{Relative residual errors of the CP and two heuristic algorithms}.}
	\label{fig.Algorithm_comparison}
\end{figure}

In Fig.~\ref{fig.Algorithm_comparison}, averaged over $1000$ realizations of the equivalent channel matrix $\check{\mathsf{H}}$, we plot the RRE of the CP, the random rotation, and the Gram-Schmidt orthogonalization algorithms for various values of the element number $n$ and the inter-element spacing $d$. 
It can be observed that, at each inter-element spacing $d=\lambda/8$, $\lambda/4$, or $\lambda/2$, the average RRE of the CP algorithm is approximately four orders of magnitude lower than those of the two heuristic algorithms when the number of elements $n$ is sufficiently large (i.e., $n=1024$). 
An interesting phenomenon is that, as the number of RIS elements increases, the performance of the CP algorithm improves to a similar level, regardless of the inter-element spacing $d$. 
This is counterintuitive, since the effect of spatial correlation becomes more pronounced at smaller inter-element spacings.
This phenomenon may be attributed to the expansion of the search space of the CP algorithm when the rank of  the equivalent channel matrix remains fixed while the number of RIS elements increases, thereby \crh{facilitating the search of two nearly-collinear} column vectors in each residual matrix $\mat{A}^{(t)}$ for $t\in [\tau]$. 
The robustness and superiority of the CP algorithm are confirmed by the above numerical results, with its performance advantage becoming particularly prominent in the scenario with a large-scale  RIS.

\subsection{Error Performance of GAM Transceivers}
In this subsection, we turn to the error performance of the GAM transceiver with the annular constellation, using the QR-SIC transceiver with the PSK modulation proposed in \cite{chen2024twc} as the benchmark. 
For implementation simplicity, we only focus on $\tau-1$ phase-modulated subchannels~\eqref{eq:subchannel}, and the phase-rotated input signal $\check{\mathbf{X}}$ to be demodulated from the beamforming subchannel~\eqref{eq:beamforming_subchannel} is assumed to have unit norm, i.e., $\left\|\check{\mathbf{X}}\right\|_2=1$. 
Besides, the inter-element spacing $d$ of the used RIS is fixed to be $\lambda/8$.

In Fig.~\ref{Fig.Received_constellation}, for a randomly-chosen channel, \review{we plot the annular constellation, the PSK constellation,} and the corresponding received constellation points under the condition $\snr=49$~dB. For a fair comparison, we let three annular and  three PSK constellations have the same received MED, which is appropriately chosen according to the SNR. 
The corresponding constellation parameters are summarized in Table~\ref{tab:simulation_parameters}, where the modulation order is expressed in bits per symbol (bps). 
It can be seen from Fig.~\ref{Fig.Received_constellation} that the radius of the PSK constellation for the $i$-th subchannel decreases as $i$ increases, since, by the convention of the SVD, the main diagonal element $r_{i,i}$ of the coefficient matrix $\mat{R}$ is monotonically decreasing. Table~\ref{tab:simulation_parameters} shows that the proposed GAM transceiver achieves a $41.3\%$ higher total transmission rate than the QR-SIC transceiver while maintaining the same MED.

\begin{figure}[!htp]
	\centering
	\subfigure[Subchannel 1 (GAM)]
	{\label{Fig.SER h1}\includegraphics[width=0.24\textwidth]{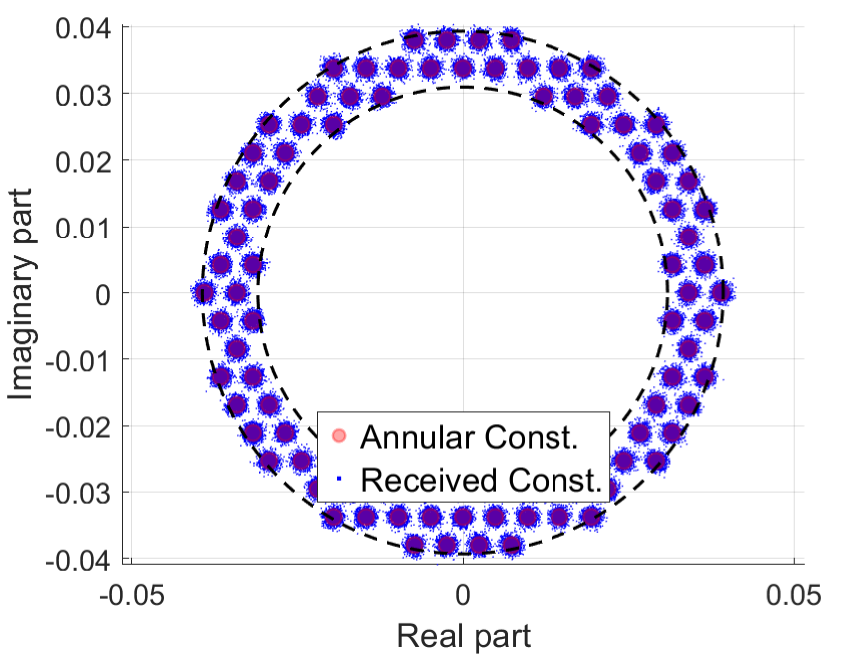}}
	\subfigure[Subchannel 1 (QR-SIC)]
	{\includegraphics[width=0.24\textwidth]{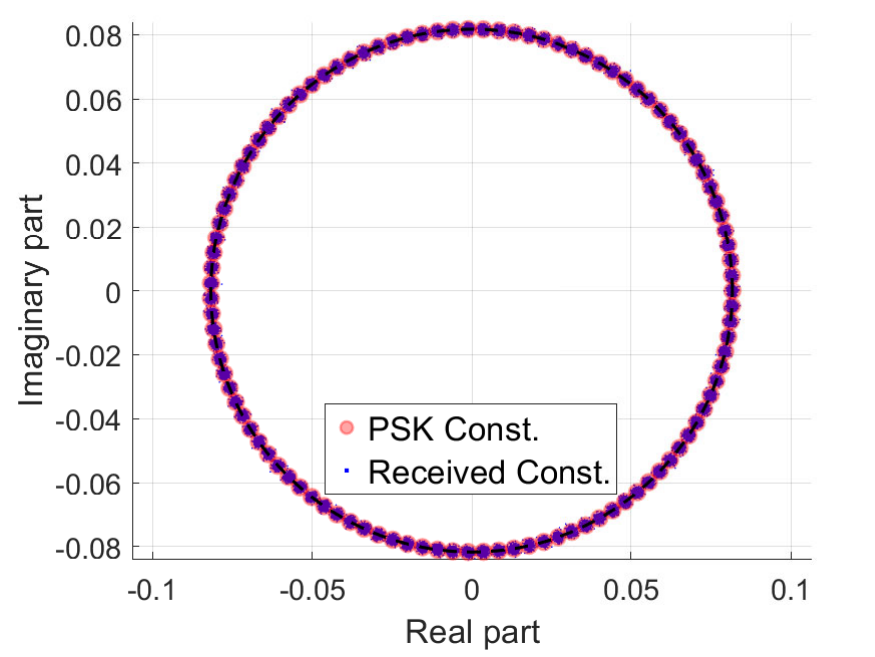}}
	\subfigure[Subchannel 2 (GAM)]
	{\includegraphics[width=0.24\textwidth]{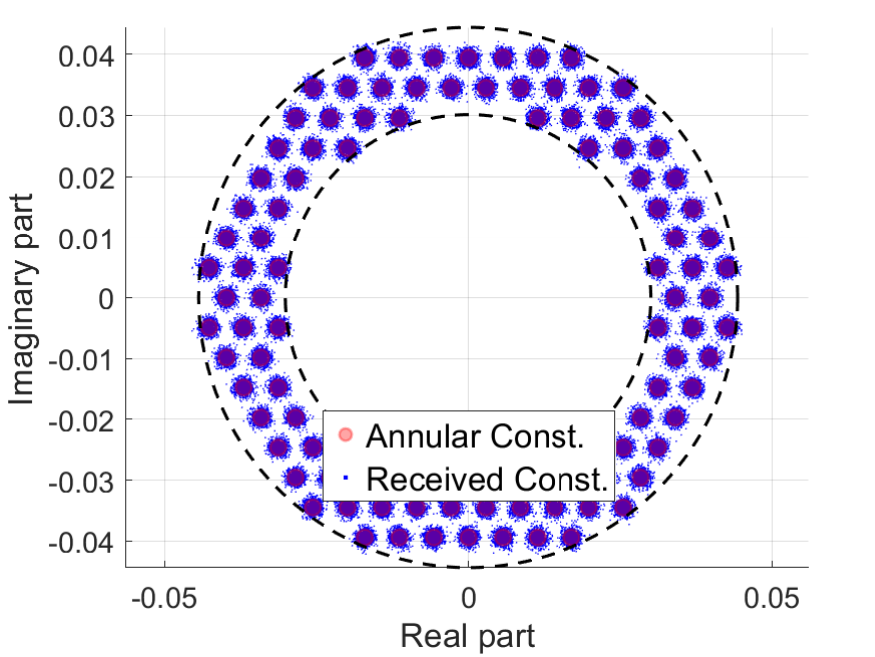}}
	\subfigure[Subchannel 2 (QR-SIC)]
	{\includegraphics[width=0.24\textwidth]{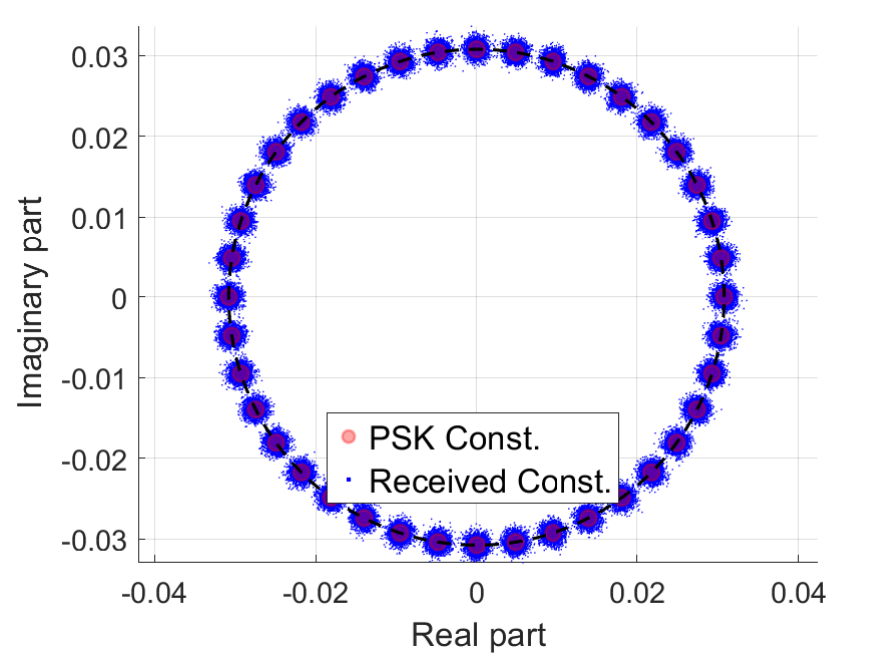}}
	\subfigure[Subchannel 3 (GAM)]
	{\includegraphics[width=0.24\textwidth]{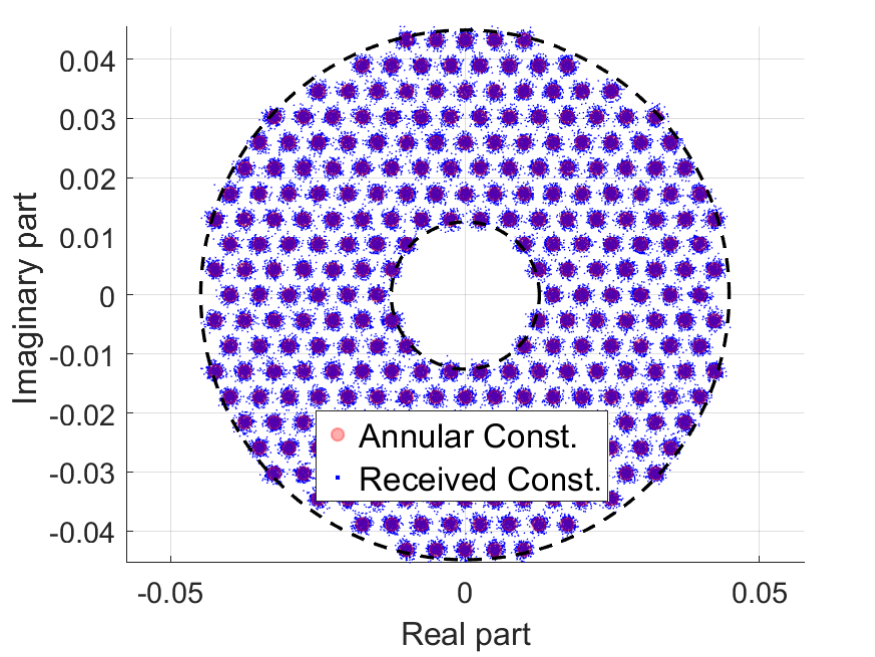}}
	\subfigure[Subchannel 3 (QR-SIC)]
	{\includegraphics[width=0.24\textwidth]{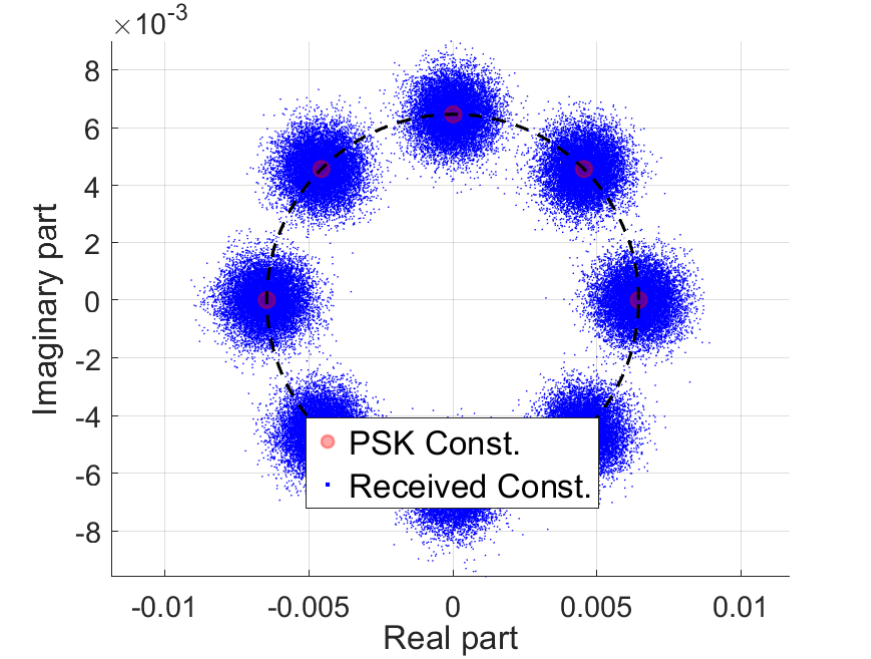}}
	\caption{Constellation points of the GAM and the QR-SIC transceivers.}
	\label{Fig.Received_constellation}
\end{figure}

\begin{table}[t]
	\centering
	\caption{Constellation parameters of the GAM and the QR-SIC transceivers.}
	\label{tab:simulation_parameters}
	\begin{tabular}{ccccc}
		\toprule
		\multirow{2}{*}{\rule{0pt}{12pt} Subchannel} &  \multicolumn{2}{c}{Annular Constellation} & 
		\multicolumn{2}{c}{PSK Constellation}\\
		\cmidrule{2-3} \cmidrule{4-5}
		  & Cardinality & Modul. Order & Cardinality & Modul. Order
		 \\
		\midrule
        1 & 90 &  6.49 bps & 107 & 6.74 bps \\
		2 & 102 & 6.67 bps & 40 & 5.32 bps\\
		3 & 278 & 8.12 bps & 8 & 3 bps\\
		Total & 2552040 & 21.28 bps & 34240&  15.06 bps\\
		\bottomrule
	\end{tabular}
\end{table}

In Fig.~\ref{Fig.weighted_SER}, we plot uncoded SERs of individual phase-modulated subchannels for the GAM and the QR-SIC transceivers with the constellations illustrated in Fig.~\ref{Fig.Received_constellation}. 
For theoretical reference, we use an exact integral expression \cite[Eq. {(4.3-12)}]{Proakis2008} to evaluate SERs of the QR-SIC transceiver, which uses the PSK constellation in each phase-modulated subchannel. For the annular constellation, its theoretical SER is approximated by Eq.~\eqref{eq:APSK}.\review{To ensure the generality of the results, each SER curve in Fig.~\ref{Fig.average_weighted_SER} is obtained by averaging over 1000 realizations of the equivalent channel matrix $\check{\mathsf{H}}$.	
Although slight deviations can be observed due to RRE, the overall performance trend remains consistent with that in Fig.~\ref{Fig.weighted_SER}.} 
Under this condition, the proposed GAM transceiver achieves a 45.89\% higher total transmission rate than the QR-SIC transceiver.
It is noteworthy that the SER performance of the GAM transceiver is slightly worse than that of the QR-SIC transceiver. 
This is because the annular constellation has a larger error coefficient (i.e., the number of nearest neighboring constellation points) as compared to the PSK constellation. 
The small gap will diminish with increasing SNR, as the impact of the error coefficient becomes negligible in the high-SNR regime.

\section{Conclusion}
\label{Sec.Conclusion}
To address the limitation of existing RIS-assisted symbiotic communication systems in achieving full DoFs under spatially correlated channels, we proposed a novel transceiver architecture by introducing a special matrix decomposition method. To the authors' best knowledge, the proposed architecture has the potential to achieve this objective for the first time.
Based on the hexagonal lattice, we proposed an annular constellation design, which perfectly matches the constellation geometry of the proposed transceiver.
Simulation results demonstrated that the proposed transceiver has significantly higher spectral efficiency than the existing transceiver architecture under the condition of comparable error performance. 
In addition, we also provided  a theoretical performance evaluation for the proposed transceiver, which shows good agreement with the numerical results.
Future work will focus on developing more practical modeling approaches for RISs and designing more flexible and higher-performance transceiver architectures, aiming to further improve system performance and enhance applicability in real-world scenarios.

\begin{figure}[!htp]
	\centering
	\includegraphics[width=0.95\columnwidth]{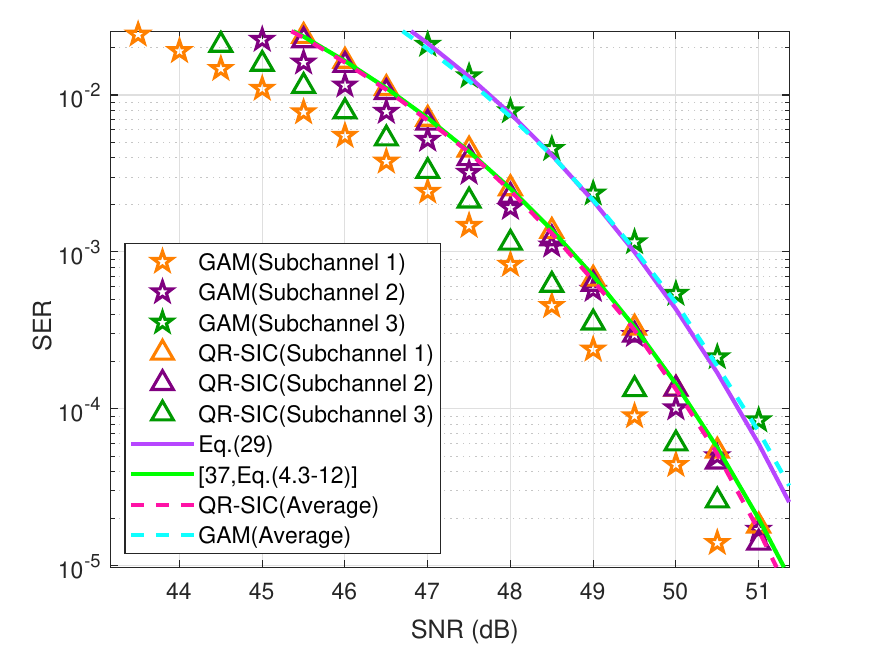}
	\caption{SER curves of the GAM and the QR-SIC transceivers with constellations illustrated in Fig.~\ref{Fig.Received_constellation}.}
	\label{Fig.weighted_SER}
\end{figure}

%

\begin{figure}[!t]
	\centering
	\includegraphics[width=0.95\columnwidth]{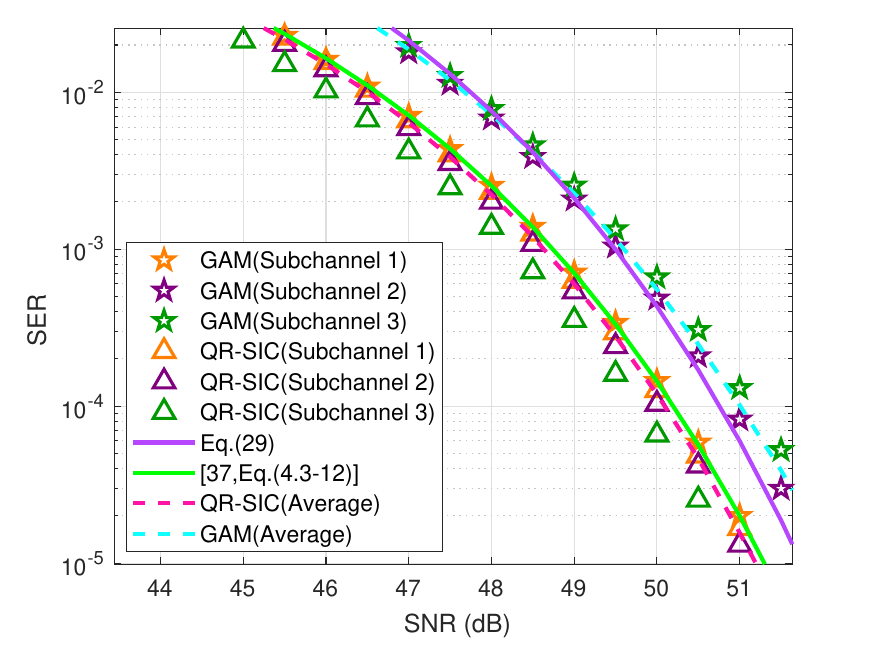}
	\caption{\review{SER curves of the GAM and the QR-SIC transceivers averaged over $1000$ channel realizations.}}
	\label{Fig.average_weighted_SER}
\end{figure}

\bibliographystyle{IEEEtran}
\bibliography{library}


\end{document}